\documentclass[11pt]{article}

\usepackage[letterpaper,hmargin=1in,vmargin=1in]{geometry}

\usepackage{times}

\usepackage{amssymb,amsmath,stmaryrd}
\usepackage{color}
\usepackage[lined,boxed,commentsnumbered]{algorithm2e}
\newtheorem{dfn}{Definition}[section]
\newtheorem{theorem}{Theorem}[section]
\newtheorem{lemma}[theorem]{Lemma}
\newtheorem{claim}[theorem]{Claim}

\newcommand{\sq}{\hbox{\rlap{$\sqcap$}$\sqcup$}}
\newcommand{\qed}{\hspace*{\fill}\sq}
\newenvironment{proof}{\noindent {\bf Proof.}\ }{\qed\par\vskip 4mm\par}

\def\N{{\mathbb{N}}}

\def\reals{{\mathbb{R}}}

\def\BR{{\mathcal{BR}}}

\newcommand{\vo}{\mathbf{0}}

\title{\bf
Efficient computation of approximate pure Nash equilibria\\in
congestion games\thanks{This work was partially supported by the grant NRF-RF2009-08 ``Algorithmic
aspects of coalitional games'' and the EC-funded STREP Project FP7-ICT-258307 EULER.}}

\author{Ioannis Caragiannis\thanks{Research Academic Computer Technology Institute \& Department of Computer Engineering and Informatics, University of Patras, 26500 Rio, Greece. Email: {\tt caragian@ceid.upatras.gr}} \and Angelo Fanelli\thanks{Division of Mathematical Sciences, School of Physical and Mathematical Sciences, Nanyang Technological University, Singapore. Email: {\tt angelo.fanelli@ntu.edu.sg, ngravin@pmail.ntu.edu.sg, askopalik@ntu.edu.sg}}\\ \and Nick Gravin$^\ddag$\\ \and Alexander Skopalik$^\ddag$}

\begin{document}

\maketitle

\begin{abstract}
Congestion games constitute an important class of games in which computing an exact or even approximate pure Nash equilibrium is in general {\sf PLS}-complete. We present a surprisingly simple polynomial-time algorithm that computes $O(1)$-approximate Nash equilibria in these games. In particular, for congestion games with linear latency functions, our algorithm computes $(2+\epsilon)$-approximate pure Nash equilibria in time polynomial in the number of players, the number of resources and $1/\epsilon$. It also applies to games with polynomial latency functions with constant maximum degree $d$; there, the approximation guarantee is $d^{O(d)}$. The algorithm essentially identifies a polynomially long sequence of best-response moves that lead to an approximate equilibrium; the existence of such short sequences is interesting in itself. These are the first positive algorithmic results for approximate equilibria in non-symmetric congestion games. We strengthen them further by proving that, for congestion games that deviate from our mild assumptions, computing $\rho$-approximate equilibria is {\sf PLS}-complete for any polynomial-time computable $\rho$.
\end{abstract}

\newpage
\section{Introduction}
\label{sec:intro}
Among other solution concepts, the notion of the pure Nash equilibrium plays a central role in Game Theory. It characterizes situations with non-cooperative deterministic players in which no player has any incentive to unilaterally deviate from the current situation in order to achieve a higher payoff. Questions related to their existence and efficient computation have been extensively addressed in the context of congestion games. In these games, pure Nash equilibria are guaranteed to exist through potential function arguments: any pure Nash equilibrium corresponds to a local minimum of a potential function. Unfortunately, this proof of existence is inefficient and computing a local minimum for this function is a computationally-hard task. This statement has been made formal in the work of Fabrikant et al. \cite{FabrikantPT04} where it is proved that the problem of computing a pure Nash equilibrium is {\sf PLS}-complete.

Such negative complexity results significantly question the importance of pure Nash equilibria as solution concepts that characterize the behavior of rational players. Approximate pure Nash equilibria, which characterize situations where no player can {\em significantly improve} her payoff by unilaterally deviating from her current strategy, could serve as alternative solution concepts\footnote{Actually, approximate pure Nash equilibria may be more desirable as solution concepts in practical decision making settings since they can accommodate small modeling inaccuracies due to uncertainty (e.g., see the arguments in \cite{CGC04}).} provided that they can be computed efficiently. In this paper, we study the complexity of computation of approximate pure Nash equilibria in congestion games and prove the first positive algorithmic results for important (and quite general) classes of congestion games. Our main result is a polynomial-time algorithm that computes $O(1)$-approximate pure Nash equilibria in congestion games under mild restrictions.

\medskip\noindent{\bf Problem statement and related work.}
Congestion games were introduced by Rosenthal \cite{R73}. In a congestion game, players compete over a set of resources. Each resource incurs a latency to all players that use it; this latency depends on the number of players that use the resource according to a resource-specific, non-negative, and non-decreasing latency function. Among a given set of strategies (over sets of resources), each player aims to select one selfishly, trying to minimize her individual total cost, i.e., the sum of the latencies on the resources in her strategy. Typical examples include network congestion games where the network links correspond to the resources and each player has alternative paths that connect two nodes as strategies. Congestions games in which players have the same set of available strategies are called symmetric.

Rosenthal \cite{R73} proved that congestion games admit a potential function with the following remarkable property: the difference in the potential value between two states (i.e., two snapshots of strategies) that differ in the strategy of a single player equals to the difference of the cost experienced by this player in these two states. This immediately implies the existence of a pure Nash equilibrium. Any sequence of improvement moves by the players strictly decreases the value of the potential and a state corresponding to a local minimum of the potential will eventually be reached; this corresponds to a pure Nash equilibrium. Monderer and Shapley \cite{Monderer96} proved that any game that admits such a cost-revealing (or exact) potential function is isomorphic to a congestion game.

The existence of a potential function allows us to view the problem of computing a pure Nash equilibrium as a local search problem \cite{Johnson88}, i.e., as the problem of computing a local minimum of the potential function. Fabrikant et al. \cite{FabrikantPT04} proved that the problem is {\sf PLS}-complete (informally, as hard as it could be given that there is an associated potential function). This negative result applies to symmetric congestion games as well as to non-symmetric network congestion games. Ackermann et al. \cite{AckermannRV08} studied the impact of combinatorial structure of congestion games to complexity and extended such negative results to games with linear latency functions. One consequence of {\sf PLS}-completeness results is that almost all the states of the game are such that any sequence of players' improvement moves that originates from these states must be exponentially long (in terms of the number of players) in order to reach a pure Nash equilibrium. Efficient algorithms are known only for special cases. For example, in symmetric network congestion games, Fabrikant et al. \cite{FabrikantPT04} show that the Rosenthal's potential function can be (globally) minimized efficiently by a flow computation.

The above negative results have led to the study of the complexity of approximate Nash equilibria. A $\rho$-approximate pure Nash equilibrium is a state, from which no player has an incentive to deviate so that she decreases her cost by a factor larger than $\rho$. Skopalik and V\"ocking \cite{SkopalikV08} show that, in general, the problem is still {\sf PLS}-complete for any polynomially computable $\rho$. Efficient algorithms are known only for special cases. For symmetric congestion games, Chien and Sinclair \cite{ChienS07} prove that the $(1+\epsilon)$-improvement dynamics converges to a $(1 + \epsilon$)-approximate Nash equilibrium after a polynomial number of steps; this result holds under additional mild assumptions on the latency functions (a ``bounded jump'' property) and the participation of the players in the dynamics. Skopalik and V\"ocking \cite{SkopalikV08} prove that this approach cannot be generalized.
They present non-symmetric congestion games with latency functions satisfying the bounded-jump property, so that every sequence of approximate improvement moves from a given initial state to an approximate equilibrium is exponentially long. Daskalakis and Papadimitriou \cite{DP07} present algorithms for the broader class of anonymous games assuming that the number of players' strategies is constant; for congestion games, this assumption is very restrictive. Efficient algorithms for approximate equilibria have been recently obtained for other classes of games such as constraint satisfaction \cite{BhalgatCK10,NT09}, network formation \cite{AC09}, and facility location games \cite{CH10}.

In light of these negative results, several authors have considered other properties of the dynamics of congestion games. The papers \cite{AwerbuchAEMS08,FanelliFM08} consider the question of whether efficient states (in the sense that the total cost of the players, or social cost, is small compared to the optimum one) can be reached by best-response moves. Recall that such states are not necessary approximate Nash equilibria. Fanelli et al. \cite{FanelliFM08} proved that congestion games with linear latency functions converge to states that approximate the optimal social cost within a constant factor after an almost linear (in the number of players) number of best response moves under mild assumptions on the participation of each player in the dynamics. Negative results in \cite{AwerbuchAEMS08} indicate that these assumptions are necessary in order to obtain convergence in subexponential time. However, Awerbuch et al. \cite{AwerbuchAEMS08} show that using almost unrestricted sequences of ($1+\epsilon$)-improvement best-response moves in congestion games with polynomial latency functions, the players rapidly converge to efficient states. Similar approaches have been followed in the context of other games as well, such as multicast \cite{CKM+08,CCL+07}, cut \cite{CMS06}, and valid-utility games \cite{MV04}.

A notion that is historically related to congestion games (but rather loosely connected to our work) is that of the price of anarchy, introduced by Koutsoupias and Papadimitriou \cite{KP99}. The price of anarchy captures the impact of selfishness on efficiency and is defined as the worst-case ratio of the social cost in any pure Nash equilibrium and the social optimum (see \cite{Roughgarden09} and the references therein for tight bounds on congestion games). Christodoulou et al. \cite{ChristodoulouKS09} extended the notion of the price of anarchy to approximate equilibria and provided tight bounds for congestion games with polynomial latency functions.

\medskip\noindent{\bf Our contribution.} We present the first polynomial-time algorithm that computes $O(1)$-approximate pure Nash equilibria in {\em non-symmetric} congestion games with polynomial latency functions of constant maximum degree. In particular, our algorithm computes $(2+\epsilon)$-approximate pure Nash equilibria in congestion games with linear latency functions, and $d^{O(d)}$ approximate equilibria for polynomial latency functions of maximum degree $d$. The algorithm is surprisingly simple. Essentially, starting from a specific initial state, it computes a sequence of best-response player moves of length that is bounded by a polynomial in the number of players and $1/\epsilon$. To the best of our knowledge, the existence of such short sequences was not known before and is interesting in itself. The sequence consists of phases so that the players that participate in each phase experience costs that are polynomially related. This is crucial in order to obtain convergence in polynomial time. Another interesting part of our algorithm is that, within each phase, it coordinates the best response moves according to two different (but simple) criteria; this is the main tool that guarantees that the effect of a phase to previous ones is negligible and, eventually, an approximate equilibrium is reached. The parameters used by the algorithm and its approximation guarantee have a nice relation to properties of Rosenthal's potential function. Our bounds are marginally higher than the worst-case ratio of the potential value at an almost exact pure Nash equilibrium over the globally optimum potential value.

We remark that, following the classical definition of polynomial latency functions in the literature, we assume that they have non-negative coefficients. We show that this is a necessary limitation. In particular, by significantly extending the reduction of \cite{SkopalikV08}, we prove that the problem of computing a $\rho$-approximate equilibrium in congestion games with linear latency functions with negative offsets is {\sf PLS}-complete. This negative statement also applies to games with polynomial latency functions with non-negative coefficients and maximum degree that is polynomial in the number of players.

\medskip\noindent{\bf Roadmap.} We begin with definitions and preliminary results and observations in Section \ref{sec:prelim}. The description of the algorithm then appears in Section \ref{sec:algorithm}. The analysis of the algorithm is presented in Section \ref{sec:analysis}. We conclude with a discussion that includes the statement of our {\sf PLS}-completeness result and open problems in Section \ref{sec:open}. Due to lack of space, many proofs have been put in appendix.

\section{Definitions and preliminaries}
\label{sec:prelim}

A \emph{congestion game} ${\cal G}$ is represented by the tuple $\left(N, E, (\Sigma_u)_{u \in N}, (f_e)_{e \in E}\right)$. There is a set of $n$ {\em players} $N=\{1, 2, ..., n\}$ and a set of {\em resources} $E$. Each player $u$ has a set of available {\em strategies} $\Sigma_u$; each strategy $s_u$ in $\Sigma_u$ consists of a non-empty set of resources, i.e., $s_u\subseteq 2^E$. A snapshot of strategies, with one strategy per player, is called a {\em state} and is represented by a vector of players' strategies, e.g., $S=(s_1, s_2, ..., s_n)$. Each resource $e\in E$ has a {\em latency function} $f_e:\N\mapsto \reals$ which denotes the latency incurred to the players using resource $e$; this latency depends on the number of players whose strategies include the particular resource. For a state $S$, let us define $n_e(S)$ to be the number of players that use resource $e$ in $S$, i.e., $n_e(S)=|\{u\in N: e\in s_u\}|$. Then, the latency incurred by resource $e$ to the players that use it is $f_e(n_e(S))$. The {\em cost} of a player $u$ at a state $S$ is the total latency she experiences at the resources in her strategy $s_u$, i.e., $c_u(S)=\sum_{e\in s_u}{f_e(n_e(S))}$. We mainly consider congestion games in which the resources have polynomial latency functions with non-negative coefficients. More precisely, the latency function of resource $e$ is $f_e(x) = \sum_{k=0}^d{a_{e,k}x^k}$ with $a_{e,k}\geq 0$. The special case of linear latency functions (i.e., $d=1$) is of particular interest. Observe that for polynomials with non-negative coefficients and maximum degree $d$, we have $f_e(x+1)\leq 2^d f_e(x)$ and $f_e(x)\leq n^d f_e(1)$ for every positive integer $x$.

Players act selfishly; each of them aims to select a strategy that minimizes her cost, given the strategies of the other players. Given a state $S=(s_1, s_2, ..., s_n)$ and a strategy $s'_u$ for player $u$, we denote by $(S_{-u},s'_u)$ the state obtained from $S$ when player $u$ {\em deviates} to strategy $s'_u$. For a strategy $S$, an {\em improvement move} (or, simply, a {\em move}) for player $u$ is the deviation to any strategy $s'_u$ that (strictly) decreases her cost, i.e., $c_u(S_{-u},s'_u) < c_u(S)$. For $q\geq 1$, such a move is called a $q$-{\em move} if it satisfies $c_u(S_{-u},s'_u) < \frac{c_u(S)}{q}$. A {\em best-response move} is a move that minimizes the cost of the player (of course, given the strategies of the other players). So, from state $S$, a move of player $u$ to strategy $s_u$ is a best-response move (and is denoted by $\BR_u(S)$) when $c_u(S_{-u},s'_u) = \min_{s\in \Sigma_u}c_u(S_{-u},s)$. With some abuse in notation, we use $\BR_u(\vo)$ to denote the best-response of player $u$ assuming that no other player participates in the game.

A state $S$ is called a {\em pure Nash equilibrium} (or, simply, an {\em equilibrium}) when $c_u(S)\leq c_u(S_{-u},s'_u)$ for every player $u\in N$ and every strategy $s'_u\in \Sigma_u$. In this case, we say that no player has (any incentive to make) a move. Similarly, a state is called a $q$-{\em approximate pure Nash equilibrium} (henceforth called, simply, a $q$-{\em approximate equilibrium}) when no player has a $q$-move.

Congestion games are potential games. They admit a {\em potential function} $\Phi:\prod_u{\Sigma_u}\mapsto \reals$, defined over all states of the game, with the following property: for any two state $S$ and $(S_{-u},s'_u)$ that differ only in the strategy of player $u$, it holds that $\Phi(S_{-u},s'_u) - \Phi(S) = c_u(S_{-u},s'_u)-c_u(S)$. Clearly, local minima of the potential function corresponds to states that are pure Nash equilibria. The function $\Phi(S)=\sum_{e\in E}{\sum_{j=1}^{n_e(S)}}{f_e(j)}$ (first used by Rosenthal \cite{R73}) is such a potential function. A nice property of this particular potential function is that the potential value at a state lies between the sum of latencies incurred by the resources and the total cost of the players.

\begin{claim}\label{claim:potential-inequalities}
For any state $S$ of a congestion game with a set of players $N$, a set of resource $E$, and latency functions $(f_e)_{e\in E}$, it holds that
$$\sum_{e\in E}{f_e(n_e(S))} \leq \Phi(S)\leq \sum_{u\in N}{c_u(S)}.$$
\end{claim}
In the rest of the paper, the term potential function is used specifically for Rosenthal's potential function.

We now present a simple observation which will be used extensively in the analysis of our algorithm. Consider a sequence of moves in which only players from a subset $F$ of $N$ participate while players in $N\setminus F$ are frozen to their strategies throughout the whole sequence. We will think of this sequence as a sequence of moves in a {\em subgame} played among the players of $F$ on the resources of $E$. In this subgame, each player in $F$ has the same set of strategies as in the original game; players of $N\setminus F$ do not participate in the subgame, although they contribute to the latency of the resources at which they have been frozen. Thus, the {\em modified} latency function of resource $e$ is then $f^{F}_{e}(x)=f_e(x+t_e)$, where $t_e$ stands for the number of players of $N\setminus F$ on resource $e$. Then, it is not hard to see that the subgame is a congestion game as well. Clearly, if $f_e$ is a linear (respectively, polynomial of maximum degree $d$) function with non-negative coefficients, so is $f^F_e$ and the bound established in Lemma \ref{lem:ratio-linear} (respectively, Lemma \ref{lem:ratio-d}, see below) also holds for the subgame. From the perspective of a player in $F$, nothing changes. At any state $S$, such a player experiences the same cost in both games and therefore has the same incentive to move, regardless whether we view $S$ as a state of the original game or the subgame. However, one should be careful with the definition of the potential for the subgame (denoted by $\Phi_F$) and use the modified latency functions $f^F_e$ instead of $f_e$. Throughout the paper, for a subset of players $F\subseteq N$, we use the notation $n^F_e(S)$ to denote the number of players in $F$ that use resource $e$ at state $S$.

\begin{claim}\label{claim:potential-sum}
Let $S$ be a state of the congestion game with a set of players $N$ and let $F\subseteq N$.
Then, $\Phi(S)\leq \Phi_{F}(S)+\Phi_{N\setminus F}(S)$ and $\Phi(S)\geq \Phi_F(S)$.
\end{claim}

The approximation guarantee of our algorithm for congestion games of a particular class (e.g., with linear latency functions) is strongly related to the worst-case ratio (among {\em all} congestion games in the class) between the potential of an approximate equilibrium (the factor of approximation may be picked to be close to $1$) and the minimum potential value.
%
%
%
Below, we present upper bounds on this quantity; these upper bounds are used as parameters by our algorithm. The next lemma deals with the case of linear latency functions.

\begin{lemma}\label{lem:ratio-linear}
Consider a congestion game with linear latency functions. Let $q\in [1,2)$ and let $S$ be a $q$-approximate equilibrium. Then, $\Phi(S)\leq \frac{2q}{2-q}\Phi(S^*)$, where $S^*$ is a state of the game with minimum potential.
\end{lemma}

Our next (rather rough) bound applies to polynomial latency functions of maximum degree $d$. It is obtained by observing that the desired ratio is at most $(d+1)$ times the known upper bound of $d^{O(d)}$ for the price of anarchy of $2$-approximate equilibria \cite{ChristodoulouKS09}.

\begin{lemma}\label{lem:ratio-d}
Consider a congestion game with polynomial latency functions of maximum degree $d$, where $d\ge 2$. Let $q\in [1,2]$ and let $S$
be a $q$-approximate equilibrium. Then, $\Phi(S)/\Phi(S^*) \in d^{O(d)}$, where $S^*$ is a state of the game with the minimum potential.
\end{lemma}

\section{The algorithm}
\label{sec:algorithm}

In this section we describe our algorithm. It takes as input a
congestion game ${\cal G}$ with $n$ players and polynomial latency functions of
maximum degree $d$ and produces a state of
${\cal G}$. The algorithm uses a constant parameter $\psi>0$ and
two more parameters $q$ and $p$. Denote by $\theta_d(q)$ the upper
bound on the worst-case ratio (among all possible congestion games
with polynomial latency functions of degree $d$) between the
potential of any $q$-approximate equilibrium and the minimum
potential value that are provided by Lemmas \ref{lem:ratio-linear}
and \ref{lem:ratio-d}, i.e., $\theta_1(q)=\frac{2q}{2-q}$ and
$\theta_d(q)=d^{O(d)}$ for $d\geq 2$. We set the parameter $q$ to be
slightly larger than $1$ (in particular, $q=1+n^{-\psi}$) and
parameter $p$ to be slightly larger than $\theta_d(q)$ (in
particular, $p=\left(\frac{1}{\theta_d(q)}-n^{\psi}\right)^{-1}$).

The algorithm considers the \emph{optimistic cost} $\ell_u$ of each player $u$,
given by $\ell_u=\min_{s_u \in \Sigma_u} \sum_{e\in s_u} f_e(1)$; this is the minimum cost that $u$ could experience assuming
that no other player participates in the game. Let $\ell_{\max}$ denote the maximum optimistic cost among all players. The algorithm partitions
the players into \emph{blocks} $B_1, B_2, \ldots, B_m$; block $B_i$ contains player $u$ if and only if $\ell_u\in(b_{i+1}, b_{i}]$, where $b_i=\ell_{\max}\left(2^{d+1}n^{2\psi+d+1}\right)^{-i+1}$. It initializes each player $u$ to choose strategy $\BR_u(\vo)$.
Then, the algorithm coordinates best-response moves by the players as follows. By considering $i$ in the increasing order (from $1$ to $m-1$), it executes phase $i$ provided that block $B_i$ is non-empty. When at phase $i$, the algorithm lets players in $B_i$ make best-response $p$-moves and players in $B_{i+1}$ make best-response $q$-moves while this is possible. The algorithm is depicted in the following table.

\IncMargin{3em}
\RestyleAlgo{boxed}
\LinesNumbered
\begin{algorithm}
\SetKwData{Left}{left}\SetKwData{This}{this}\SetKwData{Up}{up}
\SetKwFunction{Union}{Union}\SetKwFunction{FindCompress}{FindCompress}
\SetKwInOut{Input}{input}\SetKwInOut{Output}{output}

\Input{A congestion game ${\cal G} = \left(N, E, (\Sigma_i)_{i \in N}, (f_e)_{e \in E}\right)$ with $n$ players and polynomial latency functions of maximum degree $d$}
\Output{A state of ${\cal G}$}

   Set $q=1+n^{-\psi}$ and $p=\left(\frac{1}{\theta_d(q)}-n^{-\psi}\right)^{-1}$\; \label{alg:step1}
   \lForEach{$u\in N$}{set $\ell_u = c_u\left(\BR_u(\vo)\right)}$\; \label{alg:step2}
   Set $\ell_{\min}=\min_{u\in N}{\ell_u}$, $\ell_{\max}=\max_{u\in N}{\ell_u}$, and set
       $m=1+\left\lceil\log_{2^{d+1}n^{2\psi+d+1}}{\left(\ell_{\max}/\ell_{\min}\right)}\right\rceil$\;  \label{alg:step3}
   (Implicitly) partition players into \emph{blocks} $B_1, B_2,\ldots, B_m$, such that $u\in B_i\Leftrightarrow
       \ell_u\in\left(\ell_{\max} \left(2^{d+1}n^{2\psi+d+1} \right)^{-i},\ell_{\max}
       \left(2^{d+1}n^{2\psi+d+1}\right)^{-i+1}\right]$\; \label{alg:step4}
   \lForEach{$u\in N$}{set $u$ to play the strategy $s_u\leftarrow \BR_u(\vo)$}\; \label{alg:step5}
   \For{phase $i\leftarrow 1$ \KwTo $m-1$ such that $B_i\not=\emptyset$ \label{main}}{
         \While{there exists a player $u$ that either belongs to $B_i$ and has a $p$-move or belongs to $B_{i+1}$ and has a $q$-move}
         {
         $u$ deviates to the best-response strategy $s_u\leftarrow\BR_u(s_1,\ldots,s_n)$.
         }
   }
\caption{Computing approximate equilibria in congestion games.}\label{alg}
\end{algorithm}
\DecMargin{3em}

We remark that step \ref{alg:step2} partitions the players into at most $n$ non-empty blocks. Then, the for-loop at lines \ref{main}-10
enumerates only phases $i$ such that $B_i$ is non-empty, i.e., it considers at most $n$ phases.

We conclude this section with two remarks that will be treated formally in the next section. First, the selection of the boundaries of each block to be polynomially-related is crucial in order to bound the number of steps. Second, but more importantly, we notice that each player in block $B_i$ does not move after phase $i$. At the end of this phase, the algorithm guarantees that none of these players has a $p$-move to make. The most challenging part of the analysis will be to show that the players do not have any $p(1+4n^{-\psi})$-move to make after any subsequent phase. In this respect, the definition of the phases, the selection of parameter $p$ and its relation to $\theta_d(q)$ play the crucial role.

\section{Analysis of the algorithm}\label{sec:analysis}
This section is devoted to proving our main result.
\begin{theorem}\label{thm:algorithm}
For every constant $\psi>0$, the algorithm computes a $\rho_d$-approximate equilibrium for every congestion
game with polynomial latency functions of constant maximum degree $d$ and $n$
players, where $\rho_1=2+O(n^{-\psi})$ and $\rho_d\in d^{O(d)}$. Moreover, the number of player moves is at most polynomial in $n$.
\end{theorem}

\begin{proof}
We denote by $S^0$ the state computed at step \ref{alg:step5} of the algorithm
(player $u$ plays strategy $\BR_u(\vo)$) and by $S^i$ the state after the execution
of phase $i$ for $i\geq 1$. Within each phase $i$, we denote by $R_i$ the set of
players that make at least one move during the phase. Recall that the players of block $B_i$
are those with optimistic cost $\ell_u\in (b_{i+1}, b_i]$ and that $b_i=2^{d+1}n^{2\psi+d+1} b_{i+1}$, for $i=1, ..., m$.

The proof of the theorem follows by a series of lemmas. The most crucial one is
Lemma \ref{lem:potential} where we show that the potential $\Phi_{R_i}(S^{i-1})$
of the subgame among the players in $R_i$ at the beginning of phase $i\geq 2$ is
significantly smaller than $b_i$. In general, players that move during phase $i$
experience cost that is polynomially related to $b_i$ and each of them decreases her cost (and,
consequently, the potential) by a quantity that is also polynomially related to
$b_i$. This argument is used in Lemma \ref{lem:complexity} (together with
Lemma \ref{lem:potential}) in order to show that the number of steps of the
algorithm is polynomial in $n$. More importantly, Lemma \ref{lem:potential} is
used in the proof of Lemma \ref{lem:approx} in order to show that players in
block $B_i$ are not affected significantly after phase $i$ (notice that players in
$B_i$ do not move after phase $i$). Using this lemma, we conclude in
Lemma \ref{lem:approx-summary} that the players are in a
$p(1+4n^{-\psi})$-approximate equilibrium after the execution of the algorithm.
The statement of the theorem then follows by taking into account the parameters
of the algorithm.

Let us warm up with the following lemma (to be used in the proof of Lemma \ref{lem:potential}) that relates the potential
$\Phi_{R_i}(S^i)$ with the latency the players in $R_i$ experience when they
make their last move within phase $i$.

\begin{lemma}\label{lem:cost-potential}
Let $c(u)$ denote the cost of player $u\in R_i$ just after making
her last move within phase $i$. Then,
\begin{eqnarray*}
\Phi_{R_i}(S^i)\leq \sum_{u\in R_i}{c(u)}.
\end{eqnarray*}
\end{lemma}

We now present the key lemma of our proof.

\begin{lemma}\label{lem:potential}
For every phase $i\geq 2$, it holds that $\Phi_{R_i}(S^{i-1}) \leq
\frac{b_{i}}{2^d n^\psi}$.
\end{lemma}

\begin{proof}
Assume the contrary, that $\Phi_{R_{i}}(S^{i-1})> \frac{b_i}{2^d
n^\psi}$. We will show that state $S^{i-1}$ would not be a
$q$-approximate equilibrium for the players in $R_{i}\cap B_{i}$,
which contradicts the definition of phase $i-1$ of the algorithm.

First observe that a player $u$ in $B_{i+1}$ is assigned to the
strategy $\BR_u(\vo)$ in the beginning of the algorithm and does not
move during the first $i-1$ phases. Hence, by the definition of the
latency functions, she does not experience a cost more than $n^d
b_{i+1}$ at state $S^{i-1}$. Hence, the potential $\Phi_{R_i\cap
B_{i+1}}(S^{i-1})$, which is upper-bounded by the total cost of
players in $R_i\cap B_{i+1}$, satisfies
\begin{eqnarray}\label{eq:light-players}
\Phi_{R_i\cap B_{i+1}}(S^{i-1}) \leq n^{d+1} b_{i+1}.
\end{eqnarray}

We now use the fact $\Phi_{R_{i}}(S^{i-1})\leq \Phi_{R_{i}\cap
B_{i}}(S^{i-1})+\Phi_{R_{i}\cap B_{i+1}}(S^{i-1})$ (see Claim 2.2), inequality
\eqref{eq:light-players}, and the assumption $\Phi_{R_{i}}(S^{i-1})>
\frac{b_i}{2^d n^\psi}$ to obtain
\begin{eqnarray}\nonumber
\Phi_{R_{i}\cap B_{i}}(S^{i-1}) &\geq &
\Phi_{R_{i}}(S^{i-1})-\Phi_{R_{i}\cap B_{i+1}}(S^{i-1})\\\nonumber
&> & \frac{b_i}{2^dn^{\psi}}-n^{d+1} b_{i+1}\\\nonumber &=&
\left(\frac{2^{d+1}n^{2\psi+d+1}}{2^dn^\psi}-n^{d+1}\right)
b_{i+1}\\\label{eq:3} &\geq & n^{\psi+d+1}b_{i+1}.
\end{eqnarray}

Further, we consider the dynamics of the subgame among the players in
$R_{i}$ at phase $i$. For each player $u$ in $R_{i}$, we denote by
$c(u)$ the cost player $u$ experiences just after she makes her last move in
phase $i$. Observe that every player $u$ in $B_{i}\cap
R_{i}$ decreases the potential of the subgame among the players of
$R_{i}$ by at least $(p-1)c(u)$ when she performs her last $p$-move.
Hence,
\begin{eqnarray}\nonumber
(p-1)\sum_{u\in R_{i}\cap B_{i}}{c(u)} &\leq &
\Phi_{R_{i}}(S^{i-1})-\Phi_{R_{i}}(S^{i}) \\\nonumber &\leq &
\Phi_{R_{i}\cap B_{i}}(S^{i-1})+\Phi_{R_{i}\cap
B_{i+1}}(S^{i-1})-\Phi_{R_{i}}(S^{i}) \\\nonumber &\leq &
\Phi_{R_{i}\cap B_{i}}(S^{i-1})+n^{d+1} b_{i+1}-\Phi_{R_{i}}(S^{i}) \\\label{eq:4} &< &
\left(1+\frac{1}{n^\psi}\right)\Phi_{R_{i}\cap
B_{i}}(S^{i-1})-\Phi_{R_{i}}(S^{i}).
\end{eqnarray}
The last three inequalities follow by Claim 2.2 and inequalities
\eqref{eq:light-players} and \eqref{eq:3}, respectively.

Furthermore, since each player $u$ in $R_i\cap B_{i+1}$ plays a
best-response during phase $i$, her cost after her last move will be
at most the cost she would experience by deviating to strategy
$\BR_u(\vo)$, which is at most $n^d b_{i+1}$. Then, the total cost of the
players of $R_i\cap B_{i+1}$ is at most $n^{d+1}b_{i+1}$. Now, using
Lemma 4.2, the last
observation, inequalities
\eqref{eq:4} and \eqref{eq:3}, we obtain
\begin{eqnarray*}
\Phi_{R_{i}}(S^{i}) &\leq & \sum_{u\in R_{i}}{c(u)}\\
&=& \sum_{u\in R_{i}\cap B_{i+1}}{c(u)} + \sum_{u\in R_{i}\cap B_{i}}{c(u)}\\
&< & n^{d+1} b_{i+1}+\frac{1}{p-1}\left(1+\frac{1}{n^\psi}\right)\Phi_{R_{i}\cap B_{i}}(S^{i-1})-\frac{1}{p-1}\Phi_{R_{i}}(S^{i})\\
&\leq & \frac{1}{p-1}\left(1+\frac{p}{n^\psi}\right)\Phi_{R_{i}\cap
B_{i}}(S^{i-1})-\frac{1}{p-1}\Phi_{R_{i}}(S^{i})
\end{eqnarray*}
which implies that
\begin{eqnarray}\label{eq:new1}
\Phi_{R_{i}}(S^{i}) & < & \left(\frac{1}{p}+\frac{1}{n^\psi}\right)\Phi_{R_{i}\cap B_{i}}(S^{i-1}).
\end{eqnarray}

Now, let $S^*$ be the state in which the players in $R_i\cap B_i$ play their strategies in $S^i$ and the players in $R_i\cap B_{i+1}$ (as well as every other player) play their strategies in $S^{i-1}$. Consider the deviation of each player $u$ in $R_i\cap B_{i+1}$ from her strategy in $S^i$ to her strategy $\BR_u(\vo)$ in $S^*$. Recall that the cost each player $u$ in $R_i\cap B_{i+1}$ experiences when playing strategy $\BR_u(\vo)$ is at most $n^db_{i+1}$ which means that the increase her deviation incurs to the potential of the subgame among the players in $R_i$ is at most $n^db_{i+1}$. Hence,
\begin{eqnarray}\label{eq:new2}
\Phi_{R_i}(S^*) &\leq & \Phi_{R_i}(S^i)+n^{d+1}b_{i+1}.
\end{eqnarray}
Now, using the fact that $\Phi_{R_i\cap B_i}(S^*) \leq  \Phi_{R_i}(S^*)$, together with inequalities (\ref{eq:new2}), (\ref{eq:new1}), and (\ref{eq:3}), we have
\begin{eqnarray*}
\Phi_{R_i\cap B_i}(S^*) &\leq & \Phi_{R_i}(S^*)\\
&\leq & \Phi_{R_i}(S^i)+n^{d+1} b_{i+1}\\
&< & \left(\frac{1}{p}+\frac{2}{n^\psi}\right) \Phi_{R_i\cap B_i}(S^{i-1})\\
&=& \frac{1}{\theta_d(q)} \Phi_{R_i\cap B_i}(S^{i-1}).
\end{eqnarray*}
The last inequality implies that the global minimum of the potential
value of the subgame among the players of $R_{i}\cap B_{i}$ (when all other players are frozen to their strategies in $S^{i-1}$) is strictly
smaller than $\frac{1}{\theta_d(q)} \Phi_{R_{i}\cap
B_{i}}(S^{i-1})$. Due to the definition of $\theta_d(q)$ and Lemmas 2.3 and 2.4, this contradicts the fact that
$S^{i-1}$ is a $q$-approximate equilibrium for the players
in $R_{i}\cap B_{i}$.
\end{proof}

We are ready to bound the number of best-response moves. As a matter of fact,
our upper bound is dominated by the number of best-response moves in the very first phase
of the algorithm. We remark that a weaker result could be obtained without
resorting to Lemma \ref{lem:potential}.

\begin{lemma}\label{lem:complexity}
The algorithm terminates after at most $O\left(n^{5\psi+3d+3}\right)$ best-response moves.
\end{lemma}

\begin{proof}
We will upper-bound the total number of moves during the execution of the algorithm.
After the first $n$ best-response moves in line \ref{alg:step5}, the number of phases executed by the algorithm is at most $n$. At the beginning of the first phase, the latency of any player in $R_1$ is at most $n^d b_1$ (due to the definition of block $B_1$ and of the latency functions). Hence, $\Phi_{R_1}(S^0)\leq \sum_{u\in R_1}{c_u(S^0)}\leq n^{d+1} b_1$. The minimum latency experienced by any player in $R_1$ is at least $b_3$, so each move in this step decreases the potential $\Phi_{R_1}$ by at least $(q-1)b_3$. So the total number of moves is at most $\frac{n^{d+1} b_1}{(q-1)b_3} = 2^{2d+2}n^{5\psi+3d+3}$.

At the beginning of any other phase $i\geq 2$, we have that $\Phi_{R_i}(S^{i-1})\leq \frac{b_i}{2^d n^\psi}$ (by Lemma \ref{lem:potential}). The minimum latency experienced by any player in $R_i$ is at least $b_{i+2}$, so each move in this step decreases the potential $\Phi_{R_i}$ by at least $(q-1)b_{i+2}$. So the total number of moves is at most $\frac{b_i}{2^d n^\psi (q-1)b_3} = 2^{d+2}n^{4\psi+2d+2}$.

In total, we have $O\left(n^{5\psi+3d+3}\right)$ best-response moves.
\end{proof}

The proof of the following lemma strongly relies on Lemma \ref{lem:potential}.
Intuitively, Lemma \ref{lem:potential} implies that the cost experienced
by any player of $R_i$ while moving during phase $i$ is considerably lower than
the cost of players in blocks $B_{1},\ldots, B_{i-1}$ (who are not supposed to move anymore). The latter means that, for every player $u$ in
$B_{1},\ldots, B_{i-1}$, after phase $i$, neither the cost of $u$ may increase
considerably, nor the cost that $u$ could experience by a possible deviation may decrease
considerably.


\begin{lemma}\label{lem:approx}
Let $u$ be a player in the block $B_t$, where $t\leq m-2$. Let
$s_{u}'$ be a strategy different from the one assigned to $u$ by the
algorithm at the end of phase $t$. Then, for each phase $i\geq t$,
it holds that
\begin{eqnarray*}
c_u(S^i) &\leq & p\cdot
c_u(S^i_{-u},s'_u)+\frac{p+1}{n^\psi}\sum_{k=t+1}^i{b_k}.
\end{eqnarray*}
\end{lemma}

\begin{proof}
We will prove the lemma using induction on $i$. For $i=t$, the
claim follows by the definition of phase $i$ of the algorithm.
Assume that the claim is true for a phase $i$ with $t\leq i\leq
m-2$. In the following, we show that the claim is true for the phase
$i+1$ as well.

First, we show that if
\begin{eqnarray}\label{eq:1}
c_u(S^{i+1}) &\leq & c_u(S^{i})+\frac{b_{i+1}}{n^{\psi}}
\end{eqnarray}
and
\begin{eqnarray}\label{eq:2}
c_u(S^{i}_{-u},s'_u) &\leq &
c_u(S^{i+1}_{-u},s'_u)+\frac{b_{i+1}}{n^{\psi}}
\end{eqnarray}
then the claim holds. By the hypothesis of induction, we have
\begin{eqnarray*}
c_u(S^{i}) &\leq & p\cdot
c_u(S^{i}_{-u},s'_u)+\frac{p+1}{n^\psi}\sum_{k=t+1}^i{b_k}.
\end{eqnarray*}
Combining the above three inequalities, we obtain that

\begin{eqnarray*}
c_u(S^{i+1}) &\leq & c_u(S^i)+\frac{b_{i+1}}{n^{\psi}}\\
&\leq & p\cdot c_u(S^{i}_{-u},s'_u)+\frac{p+1}{n^\psi}\sum_{k=t+1}^i{b_k}+\frac{b_{i+1}}{n^{\psi}}\\
&\leq & p\cdot
c_u(S^{i+1}_{-u},s'_u)+\frac{p+1}{n^\psi}\sum_{k=t+1}^{i+1}{b_k},
\end{eqnarray*}
as desired.

In order to complete the proof of the inductive step we are left to prove \eqref{eq:1} and \eqref{eq:2}. We do so by proving that if one of these two inequalities does not hold, this would violate the statement of Lemma \ref{lem:potential}.

%
%
%

Assume that \eqref{eq:1} does not hold, i.e., $c_u(S^{i+1}) >
c_u(S^{i})+\frac{b_{i+1}}{n^{\psi}}$ for some player $u$ of block $B_t$, where $t\le i$.
We will show that the potential $\Phi_{R_{i+1}}(S^{i+1})$ at state $S^{i+1}$ of the subgame among the players in $R_{i+1}$ is larger than $\frac{b_{i+1}}{2^dn^\psi}$. Since the potential decreases during phase $i+1$, $\Phi_{R_{i+1}}(S^{i})$ should also be larger than $\frac{b_{i+1}}{2^dn^\psi}$, contradicting Lemma \ref{lem:potential}. Indeed, since player $u$ does not move during phase $i+1$, the increase in her cost from state $S^i$ to state $S^{i+1}$ implies the existence of a set of resources $C\subseteq s_u$ in her strategy with the following properties: each resource $e\in C$ is also used by at least one player of $R_{i+1}$ in state $S^{i+1}$ and, furthermore, $\sum_{e\in C}{f_e(n_e(S^{i+1}))} > \frac{b_{i+1}}{n^\psi}$. By Claim \ref{claim:potential-inequalities}, we obtain that $\Phi_{R_{i+1}}(S^{i+1})>\frac{b_{i+1}}{n^\psi}$.

Similarly, assume that \eqref{eq:2} does not hold for a player $u$ of block $B_t$ and a strategy
$s'_u$ that is different from $s_u$, the strategy assigned to $u$ in phase $t$, i.e.,
$c_u(S^{i}_{-u},s'_u) > c_u(S^{i+1}_{-u},s'_u)+\frac{b_{i+1}}{n^{\psi}}$. Recall that player $u$ does not move during phase $i+1$.
This implies that there exists a set of resources $C\subseteq s'_u$ with the following properties: each
resource $e\in C$ is used by at least one player of $R_{i+1}$ in state $S^i$ and, furthermore,
$\sum_{e\in C}{f_e(n_e(S^i_{-u},s'_u))} \geq \frac{b_{i+1}}{n^\psi}$.
Hence, by Claim \ref{claim:potential-inequalities} and the definition of the latency functions, we have $\Phi_{R+1}(S^i) \geq \sum_{e\in C}{f_e(n_e(S^i))} \geq \sum_{e\in C}{\frac{1}{2^d}f_e(n_e(S^i_{-u},s'_u))} > \frac{b_{i+1}}{2^d n^\psi}$. Again, this contradicts Lemma \ref{lem:potential}.

Hence, \eqref{eq:1} and \eqref{eq:2} hold and the proof of the
inductive step is complete.
\end{proof}

The next lemma follows easily by Lemma \ref{lem:approx}, the definition of $b_i$'s, and the definition of the last phase of the algorithm.

\begin{lemma}\label{lem:approx-summary}
The state computed by the algorithm is a $p\left(1+\frac{4}{n^{\psi}}\right)$-approximate equilibrium.
\end{lemma}

\begin{proof}
We have to show that in the state $S^{m-1}$, computed by the algorithm after the last phase, no
player has an incentive to deviate to another strategy in order to
decrease her cost by a factor of
$p\left(1+\frac{4}{n^{\psi}}\right)$. The claim is certainly true
for the players in the blocks $B_{m-1}$ and $B_m$ by the
definition of the last phase of the algorithm. Let $u$ be a player
in block $B_t$ with $t\leq m-2$ and let $s'_u$ be any
strategy different from the one assigned to $u$ by the algorithm
after phase $t$. We apply Lemma \ref{lem:approx} to player $u$. By the definition of $b_i$'s, we have
$\sum_{k=t+1}^{m}{b_k}\leq 2b_{t+1}$. Also,
$c_u(S^{m-1}_{-u},s'_u)\geq b_{t+1}$, since $u$ belongs to block
$B_t$. Hence, Lemma \ref{lem:approx} implies that
\begin{eqnarray*}
c_u(S^{m-1}) &\leq & p\cdot c_u(S^{m-1}_{-u},s'_u)+\frac{2(p+1)}{n^\psi} c_u(S^{m-1}_{-u},s'_u)\\
&\leq & p\left(1+\frac{4}{n^\psi}\right) c_u(S^{m-1}_{-u},s'_u),
\end{eqnarray*}
as desired. The last inequality follows since $p\geq 1$.
\end{proof}

By the definition of the parameters $q$ and $p$ in our algorithm, we obtain that the state computed is a $\rho_d$-approximate equilibrium with
\begin{eqnarray*}
\rho_d &\leq & \left(\frac{1}{\theta_d(q)}-n^{-\psi}\right)^{-1}\left(1+\frac{4}{n^{\psi}}\right),
\end{eqnarray*}
where $\theta_1(q)=\frac{2q}{2-q}$, $\theta_d(q)\in d^{O(d)}$ and $q=1+n^{-\psi}$. By making simple calculations, we obtain that $\rho_1\leq 2+O(n^{-\psi})$ and $\rho_d\in d^{O(d)}$. This completes the proof of Theorem \ref{thm:algorithm}.
\end{proof}

\section{Discussion and open problems}\label{sec:open}
We remark that the number of best response moves computed by our algorithm depends neither on the number of the resources nor on the number of strategies per player. In fact, our algorithm delegates to the players the computation of their best-response move; the overall running time then depends also on the time required by the players to compute a best-response move from any state of the game and (pseudo-)state $\vo$. Of course, the players are expected to be able to do this computation efficiently. 

The guarantee of our algorithm depends strongly on the fact that the latency functions have non-negative coefficients. Is this a severe limitation? We answer this question negatively in the next theorem where we prove that the problem of computing approximate equilibria is {\sf PLS}-complete for congestion games with linear latency functions that have negative offsets (but incurring non-negative latency to any player using the corresponding resource).


\begin{theorem}\label{theorem:pls}
  Finding an $\rho$-approximate equilibrium in a congestion game with
  linear laency functions with negative coefficients is {\sf PLS}-complete, for every
  polynomial-time computable $\rho>1$.
\end{theorem}

The reduction yields a congestion game in which every resource is contained in strategies of at most two players. It can also be turned into a congestion game with polynomial latency functions that have degree polynomial in $n$ (see Section \ref{sec:PLS}).

Our work reveals several open problems. The most challenging one is whether the guarantee for approximate equilibria that can be computed efficiently can be improved. For example, can we compute $(1+\epsilon)$-approximate equilibria in congestion games with linear latency functions in polynomial time for every (polynomially small) $\epsilon>0$? We believe that this is not the case and our algorithm is close to optimal in this sense. It would be very interesting to see how the best possible approximation guarantee relates to the worst-case ratio of the potential at an almost exact equilibrium over the minimum potential. Here, we point out that we have examples of congestion games for which the upper bound of $2$, provided by Lemma \ref{lem:ratio-linear}, is tight when $q$ approaches $1$. Extending this question to polynomial latencies is interesting as well. Note that a nice consequence of our work is that, besides being approximate equilibria, the states computed have low price of anarchy as well (e.g., at least $7.33+O(\epsilon)$ for linear latency functions according to the bounds in \cite{ChristodoulouKS09}). Providing improved guarantees for the social cost of approximate equilibria that can be computed efficiently or related trade-offs is another interesting line of research. Finally, we strongly believe that our techniques could be applicable to other potential games as well. Typical examples include constraint satisfaction games such as the cut and parity games studied in \cite{BhalgatCK10}; we plan to consider such games in future work.


\appendix

\section{Proofs omitted from Sections \ref{sec:prelim} and \ref{sec:analysis}}

\paragraph{Proof of Claim \ref{claim:potential-inequalities}.}
The first inequality follows easily by the definition of function
$\Phi$. The second one can be obtained by the following derivation:
\begin{eqnarray*}
\Phi(S) &=&\sum_{e\in E}\sum_{j=1}^{n_e(S)}f_e(j)\\
&\leq & \sum_{e\in E}n_e(S)\cdot f_e\Big(n_e(S)\Big)\\
&=& \sum_{u\in N}\sum_{e\in s_u}\cdot f_e\Big(n_e(S)\Big)\\
&=&\sum_{u\in N}{c_u(S)}.
\end{eqnarray*}
\qed

\paragraph{Proof of Claim \ref{claim:potential-sum}.}
We use the definition of the potential function for the original game
and the subgames, the definitions of the modified latency functions
$f^F_e(x)=f_e(x+n^{N\setminus F}_e(S))$ and $f^{N\setminus
F}_e(x)=f_e(x+n^F_e(S))$, and the equality
$n_e(S)=n^F_e(S)+n^{N\setminus F}_e(S)$ to obtain

\begin{eqnarray*}
\Phi(S) &=& \sum_{e\in E}{\sum_{j=1}^{n_e(S)}{f_e(j)}}\\
&=& \sum_{e\in E}{\sum_{j=1}^{n^F_e(S)}{f_e(j)}}+\sum_{e\in E}{\sum_{j=n^F_e(S)+1}^{n_e(S)}{f_e(j)}}\\
&\leq & \sum_{e\in E}{\sum_{j=1}^{n^F_e(S)}{f_e(j+n^{N\setminus F}_e(S))}}+\sum_{e\in E}{\sum_{j=n^F_e(S)+1}^{n_e(S)}{f_e(j)}}\\
&=& \sum_{e\in E}{\sum_{j=1}^{n^F_e(S)}{f^F_e(j)}}+\sum_{e\in E}{\sum_{j=1}^{n^{N\setminus F}_e(S)}{f_e(j+n^F_e(S))}}\\
&=& \sum_{e\in E}{\sum_{j=1}^{n^F_e(S)}{f^F_e(j)}}+\sum_{e\in E}{\sum_{j=1}^{n^{N\setminus F}_e(S)}{f^{N\setminus F}_e(j)}}\\
&=& \Phi_F(S)+\Phi_{N\setminus F}(S),
\end{eqnarray*}
as desired for the first part of the claim. For the second part, we have
\begin{eqnarray*}
\Phi(S) &=& \sum_{e\in E}{\sum_{j=1}^{n_e(S)}{f_e(j)}}\\
&\geq& \sum_{e\in E}{\sum_{j=n^{N\setminus F}_e(S)+1}^{n_e(S)}{f_e(j)}}\\
&=& \sum_{e\in E}{\sum_{j=1}^{n^F_e(S)}{f_e(j+n^{N\setminus F}_e(S))}}\\
&=& \sum_{e\in E}{\sum_{j=1}^{n^F_e(S)}{f^F_e(j)}}\\
&=& \Phi_F(S).
\end{eqnarray*}
\qed

\paragraph{Proof of Lemma \ref{lem:ratio-linear}.}
In the proof, we will need the following technical claim.
\begin{claim}\label{claim:integers}
For every non-negative integers $x,y$, it holds true $xy\leq \frac{1}{2}x^2-\frac{1}{2}x+y^2$.
\end{claim}

\begin{proof}
For $x=1$, the claim clearly holds. Otherwise, observe that $\frac{1}{2}x^2-\frac{1}{2}x\geq \frac{1}{4}x^2$.
Then $0\leq (\frac{x}{2}-y)^2=\frac{1}{4}x^2+y^2-xy\leq\frac{1}{2}x^2-\frac{1}{2}x+y^2-xy$ and claim follows.
\end{proof}

For each player $u$ we denote by $s_u$ and $s^*_u$ the strategies she uses at states $S$ and $S^*$, respectively. Using the $q$-approximate equilibrium condition, that is $c_u(S)\leq q\cdot c_u(S_{-u},s^*_u)$, we obtain
\begin{eqnarray*}
\sum_{e\in s_u}{(a_{e,1}\cdot n_e(S)+a_{e,0})} & \leq & q\sum_{e\in s^*_u}{(a_{e,1}\cdot n_e(S_{-u},s^*_u)+a_{e,0})}
\end{eqnarray*}
for each player $u\in N$. Summing over all players, we get that their total cost is
\begin{eqnarray}\nonumber
\sum_{u\in N}{\sum_{e\in s_u}{(a_{e,1}\cdot n_e(S)+a_{e,0})}} & \leq & q\sum_{u\in N}{\sum_{e\in s^*_u}{(a_{e,1}\cdot n_e(S_{-u},s^*_u)+a_{e,0})}}\\\nonumber
&\leq & q\sum_{u\in N}{\sum_{e\in s^*_u}{(a_{e,1}\cdot(n_e(S)+1)+a_{e,0})}}\\\label{eq:cost-bound}
&= & q\sum_{e \in E}{(a_{e,1}\cdot n_e(S^*)(n_e(S)+1)+a_{e,0}\cdot n_e(S^*))}
\end{eqnarray}
In the following, we use the definitions of the potential and the latency functions, the fact that $q\geq 1$, inequality \eqref{eq:cost-bound} $\bigg(n_e(S)\cdot n_e(S^*)\le\frac{1}{2}n_{e}(S)^{2}-\frac{1}{2}n_{e}(S)+n_{e}(S^*)^{2}\bigg)$, and Claim \ref{claim:integers} to obtain
\begin{eqnarray*}
\Phi(S) &=& \sum_{e\in E}{\sum_{j=1}^{n_e(S)}{f_e(j)}}\\
&=&\sum_{e\in E}{\bigg(\frac{1}{2}a_{e,1}\cdot n_e(S)^2+\frac{1}{2}a_{e,1}\cdot n_e(S)+a_{e,0}\cdot n_e(S)\bigg)}\\
&=& \sum_{e\in E}{\bigg(\frac{1}{2}a_{e,1}\cdot n_e(S)^2+\frac{1}{2}a_{e,0}\cdot n_e(S)\bigg)}+
\sum_{e\in E}{\frac{1}{2}a_{e,1}\cdot n_e(S)}+ \sum_{e\in E}{\frac{1}{2}a_{e,0}\cdot n_e(S)}\\
&\leq & \frac{1}{2}\sum_{u\in N}\sum_{e\in s_u}{\bigg(a_{e,1}\cdot n_e(S)+a_{e,0}\bigg)}+\frac{q}{2}\sum_{e\in E}{a_{e,1}n_e(S)}+\frac{q}{2}\sum_{e\in E}{a_{e,0}n_e(S)}\\
& \leq & \frac{q}{2}\sum_{e \in E}{\bigg(a_{e,1}\cdot n_e(S^*)(n_e(S)+1)+a_{e,0}\cdot n_e(S^*)\bigg)}+
\frac{q}{2}\sum_{e\in E}{a_{e,1}\cdot n_e(S)}+\frac{q}{2}\sum_{e\in E}{a_{e,0}\cdot n_e(S)}\\
&\leq & \frac{q}{2}\sum_{e\in E}{a_{e,1}\bigg(\frac{1}{2}n_e(S)^2+\frac{1}{2}n_e(S)+n_e(S^*)^2+n_e(S^*)\bigg)}\\
&& +\frac{q}{2}\sum_{e\in E}{a_{e,0}\cdot n_e(S^*)}+\frac{q}{2}\sum_{e\in E}{a_{e,0}\cdot n_e(S)}\\
&\le& \frac{q}{2}\sum_{e\in E}{\bigg(\frac{1}{2}a_{e,1}\cdot n_e(S)^2+\frac{1}{2}a_{e,1}\cdot n_e(S)+a_{e,0}\cdot n_e(S)\bigg)}\\
& & + q\sum_{e\in E}{\bigg(\frac{1}{2}a_{e,1}\cdot n_e(S^*)^2+\frac{1}{2}a_{e,1}\cdot n_e(S^*)+a_{e,0}\cdot n_e(S^*)\bigg)}\\
&=& \frac{q}{2}\Phi(S)+q\Phi(S^*),
\end{eqnarray*}
and, equivalently, $\Phi(S)\leq \frac{2q}{2-q}\Phi(S^*)$.
\qed

\paragraph{Proof of Lemma \ref{lem:ratio-d}.}
Observe that $\Phi(S)\leq \sum_{u\in N}{c_u(S)}$ (see Claim \ref{claim:potential-inequalities}). We will also show that $\Phi(S^*)\geq \frac{1}{d+1}\sum_{u\in N}{c_u(S^*)}$. The desired bound then follows by the fact that the price of anarchy of $2$-approximate equilibria is at most $d^{O(d)}$. Notice that the price of anarchy is at least $\sum_{u\in N}{c_u(S)}/\sum_{u\in N}{c_u(S^*)}$.

We will use the property $\int_0^y{f(x)} dx\leq \sum_{j=1}^y{f(j)},$ that holds for every non-decreasing function $f:[0,y]\rightarrow R$ and
integer $y\ge 1$. We prove the desired inequality as follows.


\begin{eqnarray*}
\Phi(S^*) & = & \sum_{e\in E}{\sum_{j=1}^{n_e(S^*)}{f_e(j)}}\\
&\geq &\sum_{e\in E}{\int_0^{n_e(S^*)}{f_e(x)dx}}\\
&=& \sum_{e\in E}{\int_0^{n_e(S^*)}{\sum_{k=0}^d{a_{e,k}\cdot x^k}dx}}\\
&=& \sum_{e\in E}{\sum_{k=0}^d{\frac{a_{e,k}}{k+1}n_e(S^*)^{k+1}}}\\
&\geq & \frac{1}{d+1}\sum_{e\in E}{\sum_{k=0}^d{a_{e,k}\cdot n_e(S^*)^{k+1}}}\\
&=&\frac{1}{d+1}\sum_{e\in E}n_e(S^*)f_e(S^*)\\
& = & \frac{1}{d+1}\sum_{u\in N}{c_u(S^*)}.
\end{eqnarray*}
\qed

\paragraph{Proof of Lemma \ref{lem:cost-potential}.}
We denote by $s_u$ the strategy of player $u$ at state $S^i$. We rank the players
that use resource $e$ in $S^i$ according to the timing of their last moves (using consecutive integers $1, 2, ...$). We
denote by $\mbox{rank}_e(u)$ the number of players in $R_i$ with the
smaller ranking than $u$ on resource $e$. Then, we get
$c(u)\ge\sum_{e\in s_u}{f^{R_i}_e(\mbox{rank}_e(u))}$, since any resource $e$ in $s_u$
is occupied by at least $\mbox{rank}_e(u)$ players from $R_i$ at state $S^i$: $u$ and the players with
ranks $1, 2, ..., \mbox{rank}_e(u)-1$ that made their last move before $u$. Hence, by the definition of the potential function (expressed using the modified latency functions for the subgame among the players of $R_{i}$), we have
\begin{eqnarray*}
\Phi_{R_i}(S^i) &= & \sum_{e\in E}{\sum_{j=1}^{n^{R_i}_e(S^i)}{f^{R_i}_e(j)}}\\
&=& \sum_{e\in E}~{\sum_{u\in R_i: e\in s_u}{f^{R_i}_e(\mbox{rank}_e(u))}}\\
&=& \sum_{u\in R_i}{\sum_{e\in s_u}{f^{R_i}_e(\mbox{rank}_e(u))}}\\
&\leq & \sum_{u\in R_i}{c(u)},
\end{eqnarray*}
and the lemma follows.
\qed 

%
%
%
%
%
%
%

\newcommand{\NN}{\ensuremath{\mathbb{N}}}
\newcommand{\ZZ}{\ensuremath{\mathbb{Z}}}
\newcommand{\XX}{\ensuremath{\overline{X}}}
\renewcommand{\SS}{\ensuremath{{\cal S}}}
\newcommand{\ZZZ}{\ensuremath{\mathcal{Z}}}
\newcommand{\inn}{\ensuremath{\in \{1,\ldots,n\}}}
\newcommand{\inm}{\ensuremath{\in \{1,\ldots,m\}}}
\newcommand{\inb}{\ensuremath{\in \{0,1\}}}
\newcommand{\ijb}{\ensuremath{i \inn\text{, }j \inm\text{, }b \inb}}

\newcommand{\Main}{\ensuremath{\mathrm{Main}} }
\newcommand{\Block}{\ensuremath{\mathrm{Block}} }
\newcommand{\ccheck}{\ensuremath{\mathrm{Check}} }
\newcommand{\change}{\ensuremath{\mathrm{Change}} }
\newcommand{\CCheck}{\ensuremath{\mathrm{Check}} }
\newcommand{\Change}{\ensuremath{\mathrm{Change}} }
\newcommand{\TriggerX}{\ensuremath{\mathrm{Trigger}X} }
\newcommand{\TriggerY}{\ensuremath{\mathrm{Trigger}Y} }
\newcommand{\BlockS}{\ensuremath{\mathrm{Block}S} }
\newcommand{\BlockX}{\ensuremath{\mathrm{Block}X} }
\newcommand{\Bit}{\ensuremath{\mathrm{Bit}} }
\newcommand{\TriggerDoneY}{\ensuremath{\mathrm{TriggerDone}Y} }
\newcommand{\Lock}{\ensuremath{\mathrm{Lock}} }
\newcommand{\LockGate}{\ensuremath{\mathrm{LockGate}} }
\newcommand{\LockG}{\ensuremath{\mathrm{Lock}G} }
\newcommand{\TriggerController}{\ensuremath{\mathrm{TriggerController}} }
\newcommand{\reset}{\ensuremath{\mathrm{reset}} }
\newcommand{\Reset}{\ensuremath{\mathrm{Reset}} }
\newcommand{\One}{\ensuremath{\mathrm{One}} }
\newcommand{\LockS}{\ensuremath{\mathrm{Lock}S} }
\newcommand{\TriggerLockG}{\ensuremath{\mathrm{TriggerLock}G} }
\newcommand{\ResetDoneY}{\ensuremath{\mathrm{ResetDone}Y} }
\newcommand{\TriggerUnlockG}{\ensuremath{\mathrm{TriggerUnlock}G} }


\section{Proof of Theorem \ref{theorem:pls}}
\label{sec:PLS}





We prove the theorem by reworking the reduction in~\cite{SkopalikV08}. From now on, we refer to it as the original or old construction or proof. In the following, we outline our  modifications in the original construction (in Section \ref{sec:PLS-B1}) and prove the correctness of the new one (in Section \ref{sec:PLS-B2}). Finally, we give a detailed description of the new construction in Section \ref{sec:PLS-B3}.

Recall that the original proof is a reduction from the {\sf PLS}-complete problem {\sc Flip} which is the following:
\begin{dfn}\label{def:FLIP}
An instance of the problem {\sc Flip} consists of a boolean circuit
$C$ with $n$ inputs and $m$ outputs.
A feasible solution is a bit vector  $x_1,\ldots,x_n$ and the objective value is
defined as $c(x) =\sum_{i=1}^k
y_i 2^{i-1}$ where $y$ is the output produced by $C$ with input $x$.  The neighborhood $N(x)$ of solution $x$
is the set of bit vectors $x'$ of length $n$ that differs from $x$ in one
bit. The objective is to find a local minimum.
\end{dfn}
The proof describes a transformation of $C$ into a congestion game $G(C)$ which has the property that every pure Nash equilibrium of $G(C)$ corresponds to a local optimum of $C$. Furthermore, it is ensured that every equilibrium is also an $\alpha$-approximate equilibrium for $\alpha \ge \max\{\rho,2\}$ by ensuring that every strategy change of a player decreases her latency by a factor of at least $\alpha$.

Our new construction has the additional property that  every resource is part of at most two players' strategies.
Therefore, it suffices to specify only the latency values for one and two players for each resource.
For the sake of readability, we depict latency functions by the two values $a/b$, which correspond to $f_e(1)=a$ and $f_e(2) = b$. This can obviously be turned into a linear function by setting $f_e(x)= (b - a) x + 2a - b$.

To simplify the presentation, there are many latency functions with $f_e(1)=0$. However, we can set $f_e(1) = 1$ and scale
all other latency values by a factor of $|E|\alpha$. This modification does not change the players' preferences and, by choosing $\alpha \ge 2\rho$, 
the theorem still holds. Thus, the latency functions can be described as polynomials with positive coefficients. However, their degree has to be polynomial in the number of players. 

A close look at the original reduction reveals that most of the resources are used by at most two players. The only resources for which this is not the case are the resources Bit1$_k$ of the subgames $G(S)$ and all Lock resources. Unfortunately, those resources are part of the lockable circuits, the most important feature of this reduction. We will replace these resources and add new strategies and new players to the game. See Figures~\ref{figure:tablecontroller} to~\ref{figure:tableY} for a complete description of the players' strategies, the resources, and latency functions.

\begin{figure}[t]
\centerline{\fbox{\begin{minipage}{\textwidth} 
\begin{enumerate}
\item The controller switches from $\LockS_0$ to $\LockS^j_{i,b}$.
\item All players $\LockG_k$ of all circuits except $S^j_{i,b}$ move (in increasing order) to Unlock.\\
      By moving from her One-strategy to her strategy $\change^j_{i,b}$, player
$Y_j$
\begin{itemize}
\item triggers player $X_i$ to switch to One or Zero for $b=1$ or  $b=0$, respectively, and
\item triggers the players $Y_1,\ldots,Y_{j-1}$ to switch to One.
\end{itemize}
\item After all triggered actions are done and all gates with input $y_j$ are unlocked, player $Y_j$ can change to her
strategy $\ccheck^j_{i,b}$, and thereby it triggers the Lock players to lock circuit $S_0$ and the controller to move to $\LockS_0$
\item All Lock player of circuit $S_0$ switch to one of their Lock strategies in decreasing order.
\item The controller moves back to $\LockS_0$.
\item Player $Y_j$ changes to her Zero-strategy and all Lock players of all circuits  $S^j_{i,b}$ move to their Lock strategies in decreasing order.
\end{enumerate}
\end{minipage}}}
\caption{\label{figure:superstep}
Description of a superstep beginning and ending in a base state.}
\end{figure}

\subsection{The construction of $G(C)$}\label{sec:PLS-B1}

In the original construction, a gate player has two strategies which correspond to the two values of the output of that gate. Her best response is determined by the players that correspond to the inputs of this gate. This requires latency functions for the bit resources that are not linear, i.e., latency of $0$ for two players and latency of $\alpha^{2k}$ for three players. We can avoid this, by the following changes:

Instead of one strategy One, a gate player $G_i$ now has two strategies OneA and OneB. If input $a$ of her gate is $1$, her best response is not to choose OneA. If input $b$ of her gate is $1$, her best response is not to choose OneB. As before, choosing Zero is only a best response if both inputs are $1$. For gates that have $g_i$ as input, there is no difference between $G_i$ choosing OneA or OneB and the semantics of NAND is preserved by this construction.
The Lock resources that also required nonlinear latency functions in the original construction are replaced by individual copies for each gate player.


The Lock resources play a central role in the reduction. They ensure that a state is either expensive, i,.e., some player has latency of at least $M$, or the state is part of a sequence of states  that simulates an improvement step in the {\sc Flip}-instance. We call such a sequence a {\em superstep} (see Figure~\ref{figure:superstep}).

For each player that uses a Lock resource in the old game, we introduce copies in the new game that are only used by this player and a new Lock player.
A Lock player allocates the new Lock resources of a gate in such a way that she acts as a proxy between them.
The main difference compared to the old construction is that the Controller cannot lock the circuit herself by allocating the Lock resources. Instead, the new $\LockG_i$ players have to change to one of their Lock strategies. By doing so, they set the corresponding Lock resources free for the Controller, which in turn allows the Controller to move to the strategy that locks this circuit.
Additionally, we ensure that there is no gridlock. That is, the newly added Lock players switch their strategies whenever needed. For this purpose, there are two kinds of resources that are part of the Lock players strategies. These resources are $\TriggerLockG_i$ and $\TriggerUnlockG_i$. We add these resources to strategies of the controller and the $Y$ players to ensure that a sequence of improvement steps corresponding to a superstep is possible. Furthermore, we allow a Lock player to lock a gate only if its input gates are already locked. This guarantees that a gate can always change to its correct value.


In the Controller's strategies, the original Lock resources are replaced by their copies as described above. Furthermore, the strategy $\LockS_0$ contains the resources $\TriggerLockG_k$ for all gates $g_k$ of all circuits. This ensures that eventually a circuit $S^j_{i,b}$ is locked if changing bit $x_i$ to $b$ yields enough improvement to switch $y_j$ to $0$.
In contrast to the old construction, this requires  the Controller to have higher latency since the  $\TriggerLockG_k$ resources may cause latency of $\alpha^2$ for every gate in the circuits $S^j_{i,b}$. We adjust the latency functions of the resources $\LockS_0$ and $\TriggerController$ by a factor of $\beta$ to account for this.

Finally, the Controller now has two reset strategies instead of one in the old construction. This is necessary in order to ensure that the $Y$ players are actually reset to their One strategies before locking a circuit.

In the strategies of the $X$ and $Y$ players, the Bit and Lock resources are replaced like in the gate players' strategies as described above.
In addition to that, $Y$ players' strategies (besides their One strategies) contain the additional resource $\ResetDoneY_j$ that they share with the new Reset2 strategy of the Controller. Finally, the Check strategies contain the $\TriggerLockG_k$ resources of the Lock player of circuit $S_0$.

\subsection{Proof of correctness}\label{sec:PLS-B2}
We refrain from repeating the correctness proof of the original construction. Instead, we merely outline its arguments and point out the parts that have changed due to our modifications. The proof divides the set of states into several disjoint sets, $Z_0,\ldots,Z_6$, and shows that all equilibria are contained in one of them, $Z_0$. The states in $Z_0$ are called base states. These are all inexpensive states in which every $Y$ player plays Zero or One and the Controller plays $\LockS_0$. It then suffices to show that a base state is not an approximate equilibrium (i.e., there is an improving move) if the bit vector $x$ represented by the input player is not a local optimum of $C$.

\begin{lemma}[\cite{SkopalikV08}]
None of the states in $Z_1,\ldots,Z_4$ is an equilibrium.
\end{lemma}
The proof of the lemma considers the sets one after the other and shows that in each of them there is a player that has an improving move.
In our modified construction, we need to show that none of the resources that we added or modified prevents such a move.

In a state in $Z_1$ or $Z_2$,  Lock players of every circuit $S^{j'}_{i',b'}$ with $(j',i',b') \ne (j,i,b)$ and of circuit $S_0$ play Unlock in equilibrium. This is due to the fact, that their $\TriggerLockG_k$ resources are not allocated by another player. This
implies that none of the Lock resources of the players $Y_j$, $Y_{j'}$ (with $j' < j$), and $X_i$ is allocated by a Lock player. Therefore,
$Z_1$ and $Z_2$ do not contain equilibria of the modified game.
For a state in $Z_3$, observe that player $Y_j$ allocates the resources $\TriggerLockG_k$ for every gate $g_k$ of circuit $S_0$. In equilibrium, all Lock players of this circuit are playing one of their Lock strategies, which allows for the Controller to change to $\LockS_0$.
For a state in $Z_4$ the arguments of the original proof suffice.

\begin{lemma}[\cite{SkopalikV08}]
None of the states in $Z_5$ or $Z_6$ is an equilibrium.
\end{lemma}
$Z_5$ contains the states in which a player has latency of $M^2$ or more. In this case there is always a sequence of improving moves that lead to a state in $Z_6$.
For a state in $Z_6$, we distinguish between two cases depending on the strategy played by the Controller. If she is on Reset1, all Lock players have an incentive to play their Unlock strategies and all $Y$ players to play their One strategies. However, this allows the Controller to change to Reset2 to decrease her latency.
If the Controller plays Reset2 and has latency of less than $M^2$, all $Y$ players play their One strategy.
Therefore, in an equilibrium, all Lock player of circuit $C_0$ play one of their Lock strategies. This, however allows the Controller to change to $\LockS_0$.

\begin{lemma}[\cite{SkopalikV08}]
Suppose $s$ is a base state in equilibrium. Then, the bit vector $x$
represented by the input players is a local optimum of $C$.
\end{lemma}

In any base state in which $x$ is not a local minimum of $C$, all Lock players of a circuit $S^j_{i,b}$ that has output $1$ play one of their Lock strategies in equilibrium. Therefore the controller can change to $\LockS^j_{i,b}$. \qed

\subsection{Detailed description of $G(C)$}\label{sec:PLS-B3}
\label{appendix:definition}
 Recall that
$\beta = \alpha^{2K+1}$ with $K$ being the total number of gates over all
circuits and $\gamma = 2 \alpha \beta$. That is, $\alpha \ll \beta \ll \gamma
\ll M$.


\begin{figure}[h]
\begin{center}
\begin{tabular}[ht]{|l|l|l|}
  \hline
  Strategies& Resources & Latencies \\
  \hline
\hline

$\LockS_0$  &  $\Lock_0$          & $\beta$\\
            &  $\BlockS_0$        & $0/M^2$\\
     \cline{2-3}
     &\multicolumn{2}{l|}{{\it For all gates $g_k$ of all circuits $S^j_{i,b}$:}} \\
        \cline{2-3}
            &  \quad $\TriggerLockG_k$ & $0 / \alpha^2$\\
     \cline{2-3}
     &\multicolumn{2}{l|}{{\it For all gates $g_k$ of $S_0$:}} \\
        \cline{2-3}
          &   \quad $ \LockGate_k(\text{Controller})$ of $G(S_0)$  & $0/M^2$  \\

           \hline
$\LockS^j_{i,b}$ & $\TriggerController$                               & $1/\beta^2$\\
           &$\BlockS^j_{i,b}$                                            &$0/M^2$\\
           & $\mathrm{Block}Y_j$                                         & $0/M^2$ \\
       \cline{2-3}
     &\multicolumn{2}{l|}{{\it For all gates $g_k$ of $S^j_{i,b}$: }} \\
        \cline{2-3}
           &   \quad $ \LockGate_k(\text{Controller})$ of $G(S^j_{i,b})$  & $0/M^2$  \\

\hline
$\Reset1$&$\Reset1$  &$2M$\\
       \cline{2-3}
&\multicolumn{2}{l|}{{\it  For all $j \inm$: }} \\
       \cline{2-3}
&\quad  $\TriggerY_j$& $0 / 5\alpha^5 \gamma^j$\\
       \cline{2-3}
         &\multicolumn{2}{l|}{{\it For all gates $g_k$ of all circuits:}} \\
        \cline{2-3}
            &  \quad $\TriggerUnlockG_k$ & $ \alpha / \alpha^3$\\

  \hline
 $\Reset2$&$\Reset2$  &$M$\\
       \cline{2-3}
&\multicolumn{2}{l|}{{\it  For all $j \inm$: }} \\
       \cline{2-3}
&\quad  $\ResetDoneY_j$& $0 / M^5$\\
     \cline{2-3}
     &\multicolumn{2}{l|}{{\it For all gates $g_k$ of circuit $S_0$:}} \\
        \cline{2-3}
            &  \quad $\TriggerLockG_k$ & $0 / \alpha^2$\\
     \cline{2-3}

  \hline

\end{tabular}
\caption{\label{figure:tablecontroller}
Definition of the strategies of the Controller}
\end{center}
\end{figure}

\begin{figure}[h]
\begin{center}
\begin{tabular}[ht]{|l|l|l|}
  \hline
  Strategies & Resources & Latencies \\
\hline
\hline

  OneA & $\Bit1a_k$   & $0/\alpha^{2k}$   \\
      & $\Lock1a_i(G_i) $& $0/M^3$\\
      \cline{2-3}
     &\multicolumn{2}{l|}{{\it If gate $g_k$ has $g_i$ as input $a$:}}\\
        \cline{2-3}
  &  \quad $\Bit1a_k$   & $0/\alpha^{2k}$\\
  &  \quad $\Lock1a_k(G_i)$   &$0/M^3$ \\
      \cline{2-3}
     &\multicolumn{2}{l|}{{\it If gate $g_k$ has $g_i$ as input $b$:}}\\
        \cline{2-3}
  &  \quad $\Bit1b_k$   & $0/\alpha^{2k}$\\
  &  \quad $\Lock1b_k(G_i)$   &$0/M^3$ \\

  \hline

  OneB & $\Bit1b_k$   & $0/\alpha^{2k}$   \\
      & $\Lock1b_i(G_i) $& $0/M^3$\\
      \cline{2-3}
     &\multicolumn{2}{l|}{{\it If gate $g_k$ has $g_i$ as input $a$:}}\\
        \cline{2-3}
  &  \quad $\Bit1a_k$   & $0/\alpha^{2k}$\\
  &  \quad $\Lock1a_k(G_i)$   &$0/M^3$ \\
      \cline{2-3}
     &\multicolumn{2}{l|}{{\it If gate $g_k$ has $g_i$ as input $b$:}}\\
        \cline{2-3}
  &  \quad $\Bit1b_k$   & $0/\alpha^{2k}$\\
  &  \quad $\Lock1b_k(G_i)$   &$0/M^3$ \\
 \hline

  Zero & $\Bit0a_i$& $0/\alpha^{2k} $\\
       &    $\Bit0b_i$& $0/\alpha^{2k} $\\
        & $\Lock0a_i(G_i) $& $0/M^3$\\
        & $\Lock0b_i(G_i) $& $0/M^3$\\
        \cline{2-3}
  &\multicolumn{2}{l|}{ {\it If gate $g_k$ has $g_i$ as input $a$:}} \\
     \cline{2-3}
  & \quad $\Bit0a_k$ of $G(S)$ & $0 / \alpha^{2k}$\\
  & \quad $\Lock0a_k(G_i)$ of $G(S)$ &$0/M^3$ \\
   \cline{2-3}
    &\multicolumn{2}{l|}{{\it If gate $g_k$  has $g_i$ as input $b$:}}\\
   \cline{2-3}
  & \quad $\Bit0b_k$   & $0 / \alpha^{2k}$\\

  & \quad $\Lock0b_k(G_i)$   &$0/M^3$ \\
  \hline

\end{tabular}
\end{center}
\caption{\label{figure:tableG}
Definition of the strategies of the players $G_i$ with $1 \le i \le k$.}
\end{figure}

\begin{figure}[h]
\begin{center}
\begin{tabular}[ht]{|l|l|l|}
  \hline
  Strategies & Resources & Latencies \\
\hline
\hline

  One &   $\TriggerX_{i,0}$                       &$0/\alpha \beta$ \\
      &   $\BlockX_{i,1}$& $0/M^4$\\
      \cline{2-3}
     &\multicolumn{2}{l|}{{\it If gate $g_k$ of circuit $S$ has $x_i$ as input $a$:}}\\
        \cline{2-3}
  &  \quad $\Bit1a_k$ of $G(S)$  & $0/\alpha^{2k}$\\
  &  \quad $\Lock1a_k(X_i)$ of $G(S)$ &$0/M^3$ \\
      \cline{2-3}
     &\multicolumn{2}{l|}{{\it If gate $g_k$ of circuit $S$ has $x_i$ as input $b$:}}\\
        \cline{2-3}
  &  \quad $\Bit1b_k$ of $G(S)$  & $0/\alpha^{2k}$\\
  &  \quad $\Lock1b_k(X_i)$ of $G(S)$ &$0/M^3$ \\
  \hline
  Zero            & $\TriggerX_{i,1}$& $0/\alpha \beta$\\
     &   $\BlockX_{i,0}$& $0/M^4$\\
        \cline{2-3}
  &\multicolumn{2}{l|}{ {\it If gate $g_k$ of circuit $S$ has $x_i$ as input $a$:}} \\
     \cline{2-3}
  & \quad $\Bit0a_k$ of $G(S)$ & $0 / \alpha^{2k}$\\
  & \quad $\Lock0a_k(X_i)$ of $G(S)$ &$0/M^3$ \\
   \cline{2-3}
    &\multicolumn{2}{l|}{{\it If gate $g_k$ of circuit $S$ has $x_i$ as input $b$:}}\\
   \cline{2-3}
  & \quad $\Bit0b_k$ of $G(S)$  & $0 / \alpha^{2k}$\\

  & \quad $\Lock0b_k(X_i)$ of $G(S)$  &$0/M^3$ \\
  \hline
\end{tabular}
\end{center}
\caption{\label{figure:tableX}
Definition of the strategies of the players $X_i$ with $1 \le i \le n$.}
\end{figure}

\begin{figure}[h]
\begin{center}
\begin{tabular}[ht]{|l|l|l|}
  \hline
  Strategies & Resources & Latencies \\
\hline
\hline
Lock001&  $\TriggerUnlockG_i$ & $ \alpha / \alpha^3$\\
       & $\Lock0a_i(G_i)$   & $0/M^3$   \\
       & $\Lock0b_i(G_i)$   & $0/M^3$   \\
   \cline{2-3}
    &\multicolumn{2}{l|}{{\it If gate $g_i$  has $g_k$ as input $a$:}}\\
   \cline{2-3}
       & \quad $\Lock1a_i(G_k)$ incl. X and Y   & $0/M^3$   \\
       & \quad $\LockGate_k(\LockG_i)$ without X Y  & $0/M$   \\
   \cline{2-3}
    &\multicolumn{2}{l|}{{\it If gate $g_i$  has $g_l$ as input $b$:}}\\
   \cline{2-3}
       & \quad $\Lock1b_i(G_l)$   & $0/M^3$   \\
       & \quad $\LockGate_l(\LockG_i)$  & $0/M$   \\

\hline
Lock101 &$\TriggerUnlockG_i$ & $ \alpha / \alpha^3$\\
       & $\Lock0a_i(G_i)$   & $0/M^3$   \\
       & $\Lock0b_i(G_i)$   & $0/M^3$   \\

   \cline{2-3}
    &\multicolumn{2}{l|}{{\it If gate $g_i$  has $g_k$ as input $a$:}}\\
   \cline{2-3}
       & \quad $\LockGate_k(\LockG_i)$ &   $0/M$   \\
       & \quad $\Lock0a_i(G_k)$   & $0/M^3$   \\
   \cline{2-3}
    &\multicolumn{2}{l|}{{\it If gate $g_i$  has $g_l$ as input $b$:}}\\
   \cline{2-3}
       &  \quad $\Lock1b_i(G_l)$   & $0/M^3$   \\
       &  \quad $\LockGate_l(\LockG_i)$ &   $0/M$   \\

\hline
Lock011  & $\TriggerUnlockG_i$ & $ \alpha / \alpha^3$\\
        & $\Lock0a_i(G_i)$   & $0/M^3$   \\
       & $\Lock0b_i(G_i)$   & $0/M^3$   \\

   \cline{2-3}
    &\multicolumn{2}{l|}{{\it If gate $g_i$  has $g_k$ as input $a$:}}\\
   \cline{2-3}
       &  \quad $\Lock1a_i(G_k)$   & $0/M^3$   \\
       & \quad  $\LockGate_k(\LockG_i)$ &  $0/M$   \\
   \cline{2-3}
    &\multicolumn{2}{l|}{{\it If gate $g_i$  has $g_l$ as input $b$:}}\\
   \cline{2-3}
       & \quad $\Lock0b_i(G_l)$   & $0/M^3$   \\

       & \quad $\LockGate_l(\LockG_i)$ &   $0/M$   \\

\hline
Lock110   &$\TriggerUnlockG_i$ & $ \alpha / \alpha^3$\\
        & $\Lock1a_i(G_i)$   & $0/M^3$   \\
       &  $\Lock1b_i(G_i)$   & $0/M^3$   \\
   \cline{2-3}
    &\multicolumn{2}{l|}{{\it If gate $g_i$  has $g_k$ as input $a$:}}\\
   \cline{2-3}
       &  \quad $\Lock0a_i(G_k)$   & $0/M^3$   \\
        &  \quad $\LockGate_k(\LockG_i)$   & $0/M$   \\
    \cline{2-3}
    &\multicolumn{2}{l|}{{\it If gate $g_i$  has $g_l$ as input $b$:}}\\
   \cline{2-3}
      &  \quad $\LockGate_l(\LockG_i)$  & $0/M^3$   \\
      &   \quad $\Lock0b_i(G_l)$   & $0/M$   \\

\hline
Unlock & $\LockGate_i(\text{Controller})$       & $0/M^2$\\
       & $\TriggerLockG_i$                     & $0/\alpha^2$\\
   \cline{2-3}
&\multicolumn{2}{l|}{ {\it If gate $g_k$ has $g_i$ as input $a$:}} \\
   \cline{2-3}
      &  \quad $\LockGate_i(\LockG_k)$       & $0/M$\\

  \hline
\end{tabular}
\end{center}
\caption{\label{figure:tableLockG}
Definition of the strategies of the players Lock$G_i$ with $1 \le i \le k$.}
\end{figure}
\renewcommand{\arraystretch}{0.5}

\begin{figure}[h]
\begin{center}
\begin{tabular}[th]{|l|l|l|}
  \hline
  Strategies & Resources & Latencies \\
  \hline \hline
  One & $\One_j$                        &$4\alpha^4 \gamma^j$ \\
    \cline{2-3}
       &\multicolumn{2}{l|}{{\it If gate $g_k$ of circuit $S$ has $y_j$ as input $a$: }} \\
         \cline{2-3}
      & \quad $\Bit1a_k$ of $G(S)$  & $0/\alpha^{2k}$\\
      & \quad $\Lock1a_k(Y_j)$ of $G(S)$ &$0/M^3$ \\
    \cline{2-3}
       &\multicolumn{2}{l|}{{\it If gate $g_k$ of circuit $S$ has $y_j$ as input $b$: }} \\
         \cline{2-3}
      & \quad $\Bit1b_k$ of $G(S)$  & $0/\alpha^{2k}$\\
      & \quad $\Lock1b_k(Y_j)$ of $G(S)$ &$0/M^3$ \\

  \hline
  $\change^j_{i,b}$ & $\Change_j$                & $3\alpha^3 \gamma^j$\\
  (for all     & $\BlockS_0$ & $0/M^2$\\
  $i \inn$  &  $\TriggerX_{i,b}$  & $0 / \alpha \beta$ \\
            & $\ResetDoneY_j$     & $0/M^5$ \\
     \cline{2-3}
   and $b \inb$)     &\multicolumn{2}{l|}{{\it For all $ 1 \le j' \le j$: }}\\
          \cline{2-3}
       &\quad $\TriggerY_{j'}$ & $0 / 5\alpha^5 \gamma^{j'}$  \\
   \cline{2-3}
       &\multicolumn{2}{l|}{{\it  For all vectors $(i',j',b')$ except $(i,j,b)$ }} \\
       &\multicolumn{2}{l|}{{\it  with $i' \inn$, $j' \inm$}} \\
       &\multicolumn{2}{l|}{{\it  and $b' \inb$: }} \\

         \cline{2-3}

    & \quad $\BlockS^{j'}_{i',b'}$ & $0/M^2$\\
       \cline{2-3}
       &\multicolumn{2}{l|}{{\it If gate $g_k$ of circuit $S$ has $y_j$ as input $a$: }} \\
         \cline{2-3}
      & \quad $\Bit1a_k$ of $G(S)$  & $0/\alpha^{2k}$\\
      & \quad $\Lock1a_k(Y_j)$ of $G(S)$ &$0/M^3$ \\
    \cline{2-3}
       &\multicolumn{2}{l|}{{\it If gate $g_k$ of circuit $S$ has $y_j$ as input $b$: }} \\
         \cline{2-3}
      & \quad $\Bit1b_k$ of $G(S)$  & $0/\alpha^{2k}$\\
      & \quad $\Lock1b_k(Y_j)$ of $G(S)$ &$0/M^3$ \\
  \hline
 $\ccheck^j_{i,b}$ & $\CCheck_j$ & $2\alpha^2 \gamma^j$\\
 (for all & $\TriggerY_{j}$  & $0 / 5\alpha^5 \gamma^{j}$  \\
 $i \inn$ &  $\BlockX_{i,1-b}$& $0/M^4$\\
 and $b \inb$) & $\TriggerController$&  $1/\beta^2$\\
             & $\ResetDoneY_j$     & $0/M^5$ \\

   \cline{2-3}
  &\multicolumn{2}{l|}{{\it For all $ 1 \le j' < j$: }}\\
   \cline{2-3}
   &\quad $\TriggerDoneY_{j'}$& $0/M^4$\\

    \cline{2-3}
       &\multicolumn{2}{l|}{{\it  For all vectors $(i',j',b')$ except $(i,j,b)$ }} \\
       &\multicolumn{2}{l|}{{\it  with $i' \inn$, $j' \inm$  }} \\
       &\multicolumn{2}{l|}{{\it  and $b' \inb$: }} \\

         \cline{2-3}

    & \quad $\BlockS^{j'}_{i',b'}$ & $0/M^2$\\

    \cline{2-3}
       &\multicolumn{2}{l|}{{\it If gate $g_k$ of circuit $S$ has $y_j$ as input $a$:}} \\
         \cline{2-3}
&\quad  $\Lock0a_k(Y_j)$ of $G(S)$  &$0/M^3$ \\
  & \quad $\Bit0a_k$ of $G(S)$ & $0 / \alpha^{2k}$\\
    \cline{2-3}
   &\multicolumn{2}{l|}{{\it If gate $g_k$ of circuit $S$ has $y_j$ as input $b$:}} \\
   \cline{2-3}
 & \quad  $\Bit0b_k$ of $G(S)$  & $0 / \alpha^{2k}$\\
 & \quad $\Lock0b_k(Y_j)$ of $G(S)$ &$0/M^3$ \\
     \cline{2-3}
     &\multicolumn{2}{l|}{{\it For all gates $g_k$ of circuit $S_0$:}} \\
        \cline{2-3}
            &  \quad $\TriggerLockG_k$ & $0 / \alpha^2$\\

  \hline

  Zero            & $\TriggerY_j$& $0 / 5\alpha^5 \gamma^j$\\
  & $\TriggerDoneY_j$ & $0/M^4$\\
   & $\mathrm{Block}Y_j$                                         & $0/M^2$ \\
            & $\ResetDoneY_j$     & $0/M^5$ \\

    \cline{2-3}
       &\multicolumn{2}{l|}{{\it If gate $g_k$ of circuit $S$ has $y_j$ as input $a$:}} \\
         \cline{2-3}
& \quad $\Lock0a_k(Y_j)$ of $G(S)$  &$0/M^3$ \\
  &\quad  $\Bit0a_k$ of $G(S)$ & $0 / \alpha^{2k}$\\
    \cline{2-3}
   &\multicolumn{2}{l|}{{\it If gate $g_k$ of circuit $S$ has $y_j$ as input $b$:}} \\
   \cline{2-3}
 &\quad $\Bit0b_k$ of $G(S)$  & $0 / \alpha^{2k}$\\
 &\quad $\Lock0b_k(Y_j)$ of $G(S)$ &$0/M^3$ \\

  \hline
  \end{tabular}
\end{center}
\caption{\label{figure:tableY}
Definition of the strategies of the  players $Y_j$ with $1 \le j \le m$.}
\end{figure}

\end{document}